\documentclass[runningheads]{llncs}
\usepackage[T1]{fontenc}
\usepackage{graphicx}

\usepackage{amsmath,amssymb}

\usepackage{algorithm2e}
\usepackage{booktabs}
\usepackage{listings}
\usepackage{tikz}
\usepackage{microtype}
\usepackage{xurl}
\usepackage[hidelinks]{hyperref}

\usetikzlibrary{arrows.meta,patterns,positioning,matrix,decorations.pathreplacing}

\newcommand{\prebibadjust}{0pt}
\newcommand{\quotes}[1]{``#1''}
\begin{document}

\title{The Inverse Lyndon Array:\texorpdfstring{\\}{ } Definition, Properties, and Linear-Time Construction}
\titlerunning{The Inverse Lyndon Array}

\author{Pietro Negri\inst{1} \and Manuel Sica\inst{1} \and Rocco Zaccagnino\inst{1} \and Rosalba Zizza\inst{1}}
\authorrunning{P. Negri et al.}

\institute{Department of Computer Science, University of Salerno, Fisciano (SA), Italy\\
\email{\{masica, rzaccagnino, rzizza\}@unisa.it, p.negri137@gmail.com}}
\maketitle

\begin{abstract}
The Lyndon array stores, at each position of a word, the length of the longest maximal Lyndon subword starting at that position, and plays an important role in combinatorics on words, for example in the construction of fundamental data structures such as the suffix array. 

In this paper, we introduce the \textit{Inverse Lyndon Array}, the analogous structure for inverse Lyndon words, namely words that are lexicographically greater than all their proper suffixes. Unlike standard Lyndon words, inverse Lyndon words may have non-trivial borders, which introduces a genuine theoretical difficulty. We show that the inverse Lyndon array can be characterized in terms of the next greater suffix array together with a border-correction term, and prove that this correction coincides with a longest common extension (LCE) value. Building on this characterization, we adapt the nearest-suffix framework underlying Ellert’s linear-time construction of the Lyndon array to the inverse setting, obtaining an $O(n)$-time algorithm for general ordered alphabets. Finally, we discuss implications for suffix comparison and report experiments on random, structured, and real datasets showing that the inverse construction exhibits the same practical linear-time behavior as the standard one.
\end{abstract}

\section{Introduction}

Lyndon words, introduced by Lyndon~\cite{lyndon-words} and extensively studied in combinatorics on words, are primitive words that are lexicographically minimal among their rotations. Equivalently, a non-empty word $w$ is a \emph{Lyndon word} if $w \prec s$ for every non-empty proper suffix $s$ of $w$, where $\prec$ denotes the lexicographic order induced by a total order on the alphabet~$\Sigma$.

The \emph{Lyndon array} $\lambda[1..n]$ of a word $x[1..n]$ gives, at each position $i$, the length of the longest Lyndon factor of $x$ starting at~$i$. Since Bannai et al.~\cite{runs-theorem} showed that $\lambda$ is sufficient to compute all maximal repetitions in linear time, efficient construction algorithms for $\lambda$ have received considerable attention. 
Franek et al.~\cite{algorithms-lyndon-array-original,algorithms-lyndon-array-revisited} related $\lambda$ to a \textit{next-smaller-value} computation on the inverse suffix array, while Ellert~\cite{lyndon-simple} later gave the first direct linear-time algorithm on general ordered alphabets, based on nearest smaller suffix edges and an LCE acceleration mechanism inspired by Manacher's algorithm~\cite{Manacher1975ANL}. Subsequent work considered space-efficient variants, sub-linear algorithms, and links with the Lyndon forest~\cite{space-efficient,sublinear-time,badkobeh_et_al:LIPIcs.CPM.2022.13}.

A closely related notion is that of \emph{inverse Lyndon words}, i.e., a non-empty word $w$ is inverse Lyndon if $s \prec w$ for every non-empty proper suffix~$s$~\cite{inverse-lyndon}, introduced in connection with the \emph{Inverse Canonical Lyndon Factorization} (ICFL), a variant of CFL that may yield more balanced factors and therefore well suited to bioinformatics applications~\cite{lyndon-bio-example}. A border property for the canonical inverse Lyndon factorization was introduced in~\cite{tcs2021} and later clarified in~\cite{inverse-lyndon-2}.
From a combinatorial viewpoint, the inverse setting is not a merely formal dualization of the standard one. The crucial asymmetry is that standard Lyndon words are unbordered, whereas inverse Lyndon words may have non-trivial borders. As a consequence, the classical identity $\lambda[i]=\mathit{next}[i]-i$ no longer survives unchanged. The main difficulty is therefore not merely to restate Ellert's framework with the inequalities reversed, but rather to determine precisely how borders interact with nearest greater suffixes and to show that this interaction preserves the amortized linearity of the underlying LCE machinery.
The \quotes{border correction} is structural: without identifying it as an LCE value on a nearest-greater-suffix edge, one would need separate border computation, breaking the direct linear-time reduction. Our result also connects two nearby lines of work, namely non-crossing LCE queries on general ordered alphabets~\cite{noncrossing-lce,ellert-runs-general} and the border-based view of inverse Lyndon words and ICFL~\cite{inverse-lyndon-2}. In the inverse Lyndon array setting, we show that the needed border correction is recovered directly from the LCE on the corresponding NGS edge.

\paragraph{Contributions.}
We introduce the \emph{Inverse Lyndon Array}~$\lambda^{-1}[1..n]$, which records at each position $i$ the length of the longest maximal inverse Lyndon subword starting at~$i$. The main contributions of this work are as follows:
\begin{itemize}
\item an exact characterization of $\lambda^{-1}$ through the \emph{next greater suffix} array, with a border correction term (see Lemma~\ref{lem:equiv-inv});
\item the identification of the border correction with an LCE value on a nearest-greater suffix edge, which removes the need for any explicit border computation and distinguishes our approach from border-explicit treatments in the inverse Lyndon literature (see Lemma~\ref{lem:border-lce});
\item the adaptation of the non-crossing property, chain iteration, and LCE acceleration from the classical NSS (Next Smaller Suffix)/PSS (Previous Smaller Suffix) setting to the novel NGS (Next Greater Suffix)/PGS (Previous Greater Suffix) setting, yielding a linear-time algorithm for constructing the inverse Lyndon array over general ordered alphabets;
\item a study of suffix-sorting implications, showing that the compatibility property underlying the standard Lyndon array has no direct inverse analogue, while the joint use of $\lambda$ and $\lambda^{-1}$ still yields constant-time rules for several suffix comparisons.
\end{itemize}
The implementation and benchmark scripts are available online\footnote{\url{https://github.com/FLaTNNBio/inverselyndonarray}}.

\section{Preliminaries}\label{sec:prelim}

\paragraph{Words and lexicographic orders.}
A \emph{word} $x = x[1]x[2]\cdots x[n]$ is a sequence of symbols from a totally ordered alphabet $(\Sigma, <)$. We write $|x| = n$ for its length, $\varepsilon$ for the empty word, and $\Sigma^*$ and $\Sigma^+$ for the sets of all words and all non-empty words, respectively. For $1 \le i \le j \le n$, the subword $x[i..j]$ is the sequence $x[i]x[i{+}1]\cdots x[j]$; if $i > j$, it equals $\varepsilon$. A subword $x[1..j]$ is a \emph{prefix}, and $x_i = x[i..n]$ denotes the suffix starting at position~$i$.

The \emph{lexicographic order} $\prec$ on $\Sigma^*$ is defined as follows: $x \prec y$ if and only if either $y = xw$ for some non-empty $w$, or $x = urv$ and $y = usw$ with $r < s$. The \emph{Longest Common Extension} of two suffixes is:
\[
\mathrm{lce}(i,j) = \max\{|u| \mid x_i = uv,\ x_j = uw\}.
\]

We frame every word with sentinels, i.e.,  $x = \#\,x(1..n)\,\$$, where $\#$ and $\$$ satisfy $\# < \$ < a$ for all $a \in \Sigma$ in the standard setting, and $\# > \$ > a$ for all $a \in \Sigma$ in the inverse setting. In the standard case, placing $\#$ strictly below every alphabet symbol ensures that $x_1$ is the lexicographic minimum among all suffixes of the framed word, so all suffix ranks are distinct and every NSS/PSS edge is unambiguous. In the inverse case, the reversed order makes $x_1$ the lexicographic maximum, and guarantees a strict total order on all suffixes with no ties.

A non-empty word $w$ is a \emph{Lyndon word} if $w \prec s$ for every non-empty proper suffix $s$ of~$w$~\cite{PIERREDUVAL1983363}. A subword $x[i..j]$ is a \emph{maximal Lyndon subword} of $x$ starting at~$i$ if it is a Lyndon word and either $j = n$ or $x[i..j{+}1]$ is not Lyndon.

\begin{definition}[Lyndon Array]\label{def:lyn-arr}
The \emph{Lyndon array} $\lambda[1..n]$ of a word $x[1..n]$ is defined by
\[
\lambda[i] = \max\{m \in [1,n-i+1] \mid x[i..i+m-1] \text{ is a Lyndon word}\}.
\]
\end{definition}

The suffix array $\mathit{SA}[1..n]$ is the permutation such that
\[
x_{\mathit{SA}[1]} \prec x_{\mathit{SA}[2]} \prec \cdots \prec x_{\mathit{SA}[n]}.
\]
The inverse suffix array $\mathit{ISA}$ satisfies $\mathit{ISA}[\mathit{SA}[i]] = i$, so $\mathit{ISA}[j]$ is the lexicographic rank of $x_j$. Franek et al.~\cite{algorithms-lyndon-array-original} proved the following characterization.

\begin{theorem}[\cite{algorithms-lyndon-array-original}]\label{thm:fulcro}
$x[i..j]$ is maximal Lyndon subword if and only if:
\begin{enumerate}
\item $\mathit{ISA}[i] < \mathit{ISA}[h]$ for all $h \in [i{+}1,j]$, and
\item either $j = n$ or $\mathit{ISA}[j{+}1] < \mathit{ISA}[i]$.
\end{enumerate}
\end{theorem}

Consequently, $\lambda$ can be obtained from a NSV computation on the inverse suffix array, as follows. Given an array $A[1..n]$, the NSV array is defined by:
\[
\mathrm{NSV}[i] = \min\{j>i \mid A[j] < A[i]\} - i,
\]
with $\mathrm{NSV}[i] = n-i+1$ if no such $j$ exists.

\subsection{The NSS/PSS Framework}\label{sec:nss}

We briefly recall the framework of Ellert~\cite{lyndon-simple} for computing $\lambda$ in $O(n)$ time, since it is the template for our inverse construction. Throughout this section we assume $x = \#\,x(1..n)\,\$$ with $\# < \$ < a$ for all $a \in \Sigma$.

\begin{definition}[NSS and PSS Arrays]\label{def:nss-pss}
For a word $x[1..n]$, define
\begin{align*}
\mathit{next}[i] &= \min\{j \in (i,n] \mid x_j \prec x_i\}, \\
\mathit{prev}[i] &= \max\{j \in [1,i) \mid x_j \prec x_i\},
\end{align*}
with $\mathit{next}[1] = \mathit{next}[n] = n{+}1$, $\mathit{prev}[1] = 0$, and $\mathit{prev}[n] = 1$.
\end{definition}

\begin{example}\label{ex:nss}
For $x = \#banana\$$, with $\# < \$ < a < b < n$, the arrays are:
\begin{center}
\small
$\begin{array}{l|cccccccc}
i & 1 & 2 & 3 & 4 & 5 & 6 & 7 & 8 \\
x & \# & b & a & n & a & n & a & \$ \\
\hline
\mathit{next}[i] & 9 & 3 & 5 & 5 & 7 & 7 & 8 & 9 \\
\mathit{prev}[i] & 0 & 1 & 1 & 3 & 1 & 5 & 1 & 1 \\
\lambda[i]       & 8 & 1 & 2 & 1 & 2 & 1 & 1 & 1 \\
\end{array}$
\end{center}
Figure~\ref{fig:nss-arcs} shows the corresponding NSS and PSS edges.
\end{example}

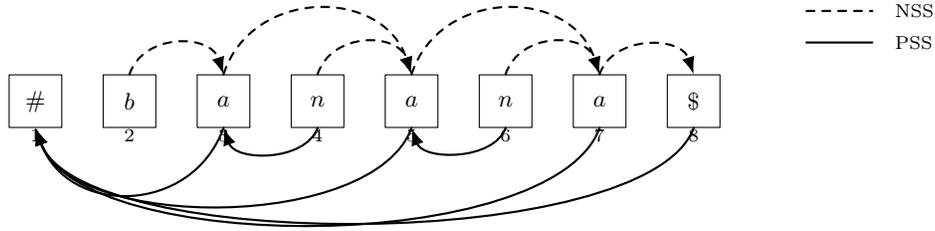
\begin{figure}[t]
\centering
\begin{tikzpicture}[x=1.25cm,y=1cm,>=Latex,line cap=round,line join=round]

  \foreach \i/\c in {1/\#,2/b,3/a,4/n,5/a,6/n,7/a,8/\$} {
    \node[draw,minimum width=7mm,minimum height=7mm,inner sep=0pt] (c\i) at (\i,0) {\footnotesize $\c$};
    \node[font=\scriptsize] at (\i,-0.45) {\i};
  }

  \draw[->,thick,dashed] (c2.north) to[out=72,in=108,looseness=1.18] (c3.north);
  \draw[->,thick,dashed] (c3.north) to[out=72,in=108,looseness=1.25] (c5.north);
  \draw[->,thick,dashed] (c4.north) to[out=72,in=108,looseness=1.22] (c5.north);
  \draw[->,thick,dashed] (c5.north) to[out=72,in=108,looseness=1.25] (c7.north);
  \draw[->,thick,dashed] (c6.north) to[out=72,in=108,looseness=1.22] (c7.north);
  \draw[->,thick,dashed] (c7.north) to[out=72,in=108,looseness=1.12] (c8.north);

  \draw[->,thick] (c3.south) to[out=-110,in=-70,looseness=1.28] (c1.south);
  \draw[->,thick] (c4.south) to[out=-110,in=-70,looseness=0.98] (c3.south);
  \draw[->,thick] (c5.south) to[out=-110,in=-70,looseness=0.75] (c1.south);
  \draw[->,thick] (c6.south) to[out=-110,in=-70,looseness=0.95] (c5.south);
  \draw[->,thick] (c7.south) to[out=-110,in=-70,looseness=0.62] (c1.south);
  \draw[->,thick] (c8.south) to[out=-110,in=-70,looseness=0.52] (c1.south);

  \draw[thick,dashed] (9.2,1.20) -- (9.9,1.20);
  \node[anchor=west,font=\scriptsize] at (10.05,1.20) {NSS};

  \draw[thick] (9.2,0.78) -- (9.9,0.78);
  \node[anchor=west,font=\scriptsize] at (10.05,0.78) {PSS};

\end{tikzpicture}
\caption{NSS and PSS edges for $x=\#banana\$$. Dashed arcs denote next smaller suffix edges, while solid arcs denote previous smaller suffix edges.}
\label{fig:nss-arcs}
\end{figure}

\begin{lemma}[\cite{lyndon-simple}]\label{lem:equiv-std}
For all $i \in [1,n]$, $\mathit{next}[i] = i + \lambda[i]$.
\end{lemma}

\begin{lemma}[\cite{lyndon-simple}]\label{lem:prev-lyndon}
For all $i \in [2,n]$, the word $x[\mathit{prev}[i]..i{-}1]$ is a Lyndon word.
\end{lemma}

The key properties enabling linear-time construction are the following.

\begin{lemma}[Non-Crossing~\cite{lyndon-simple}]\label{lem:no-crossing}
Let $l_1 < r_1$ and $l_2 < r_2$ be index pairs, each connected by an NSS or PSS edge. Then it is impossible to have $l_1 < l_2 < r_1 < r_2$.
\end{lemma}

\begin{lemma}[Chain Iteration~\cite{lyndon-simple}]\label{lem:chain}
For any $r \in [2,n]$:
\begin{enumerate}
\item $\mathit{prev}[r] = \mathit{prev}^*[r{-}1]$,
\item $\mathit{next}[l] = r$ if and only if $l = \mathit{prev}^*[r{-}1]$ and $l > \mathit{prev}[r]$,
\end{enumerate}
where $\mathit{prev}^*[r{-}1]$ denotes an element on the PSS chain from $r{-}1$.
\end{lemma}

\begin{lemma}[LCE Acceleration~\cite{lyndon-simple}]\label{lem:lce-accel}
Let $l = \mathit{prev}[k]$ and $r = \mathit{next}[k]$. Then:
\begin{enumerate}
\item if $\mathrm{lce}(l,k) = \mathrm{lce}(k,r)$, then $\mathrm{lce}(l,r) \ge \mathrm{lce}(k,r)$;
\item if $\mathrm{lce}(l,k) < \mathrm{lce}(k,r)$, then $\mathrm{lce}(l,r) = \mathrm{lce}(l,k)$ and $\mathit{prev}[r] = l$;
\item if $\mathrm{lce}(l,k) > \mathrm{lce}(k,r)$, then $\mathrm{lce}(l,r) = \mathrm{lce}(k,r)$ and $\mathit{next}[l] = r$.
\end{enumerate}
\end{lemma}

Together with the SMART-LCE procedure of~\cite{lyndon-simple}, which computes all required LCE values in amortized $O(n)$ time via a global rightmost-inspected-position variable, these lemmas yield an algorithm that computes $\lambda$ in $O(n)$ time and $O(n)$ space on general ordered alphabets.

\section{The Inverse Lyndon Array}\label{sec:inverse}

We now develop the theory of the inverse Lyndon array. Throughout this section, $x = \#\,x(1..n)\,\$$, with sentinels satisfying $\# > \$ > a$ for all $a \in \Sigma$.

\subsection{Inverse Lyndon Words}

\begin{definition}[\cite{inverse-lyndon}]\label{def:inv-lyn}
A word $w \in \Sigma^+$ is an \emph{inverse Lyndon word} if $s \prec w$ for every non-empty proper suffix~$s$ of~$w$.
\end{definition}

Unlike standard Lyndon words, inverse Lyndon words may have non-trivial borders. For example, $\mathit{dabda}$ is an inverse Lyndon word with maximal border $\mathit{da}$. The following prefix-closure property is fundamental~\cite{inverse-lyndon}.

\begin{lemma}[\cite{inverse-lyndon}]\label{lem:prefix-closed}
Every non-empty prefix of an inverse Lyndon word is an inverse Lyndon word.
\end{lemma}

\begin{definition}[Maximal Inverse Lyndon Subword]\label{def:inv-max}
A subword $x[i..j]$ is a \emph{maximal inverse Lyndon subword} starting at~$i$ if it is an inverse Lyndon word and either $j = n$ or $x[i..j{+}1]$ is not an inverse Lyndon word.
\end{definition}

\begin{definition}[Inverse Lyndon Array]\label{def:inv-arr}
The Inverse Lyndon Array is defined by:
\[
\lambda^{-1}[i] = \max\{m \in [1,n-i+1] \mid x[i..i+m-1] \text{ is an inverse Lyndon word}\}.
\]
\end{definition}

\begin{example}\label{ex:inv-arr}
Let $x = \#\mathit{aababbaa}\$$, with $\# > \$ > b > a$. The inverse Lyndon array and the associated NGS/PGS arrays are:
\begin{center}
\small
$\begin{array}{l|cccccccccc}
i & 1 & 2 & 3 & 4 & 5 & 6 & 7 & 8 & 9 & 10 \\
x & \# & a & a & b & a & b & b & a & a & \$ \\
\hline
\lambda^{-1}        & 10 & 2 & 1 & 3 & 1 & 4 & 3 & 2 & 1 & 1 \\
\mathit{next}_{-1}  & 11 & 3 & 4 & 6 & 6 & 10 & 10 & 9 & 10 & 11 \\
\mathit{prev}_{-1}  & 0  & 1 & 1 & 1 & 4 & 1 & 6 & 7 & 7 & 1 \\
\mathit{nlce}       & 0  & 1 & 0 & 1 & 0 & 0 & 0 & 1 & 0 & 0 \\
\end{array}$
\end{center}
At $i=4$, the maximal inverse Lyndon subword is $x[4..6] = \mathit{bab}$, so $\lambda^{-1}[4] = 3$. Its longest border has length $1$, hence $\mathit{next}_{-1}[4] = 6$ and $\mathit{nlce}[4] = 1$, in agreement with Corollary~\ref{cor:recover}. At $i=6$, the maximal inverse Lyndon subword is $x[6..9] = \mathit{bbaa}$, so $\lambda^{-1}[6] = 4$. Since this factor is unbordered, $\mathit{next}_{-1}[6] = 10$ and $\mathit{nlce}[6] = 0$.
\end{example}

\subsection{NGS/PGS Arrays and their Relationship to $\lambda^{-1}$}

\begin{definition}[NGS and PGS Arrays]\label{def:ngs-pgs}
For a word $x[1..n]$ with inverted sentinels, define:
\begin{align*}
\mathit{next}_{-1}[i] &= \min\{j \in (i,n] \mid x_j \succ x_i\}, \\
\mathit{prev}_{-1}[i] &= \max\{j \in [1,i) \mid x_j \succ x_i\},
\end{align*}
with $\mathit{next}_{-1}[1] = \mathit{next}_{-1}[n] = n{+}1$, $\mathit{prev}_{-1}[1] = 0$, and $\mathit{prev}_{-1}[n] = 1$.
\end{definition}

Thus, in the inverse setting, the relevant edges connect nearest \emph{greater} suffixes rather than nearest smaller suffixes.

\begin{lemma}[NGS and $\lambda^{-1}$ equivalence]\label{lem:equiv-inv}
Let $z$ be the maximal inverse Lyndon subword starting at $i$, and let $\mathrm{border}(z)$ denote the length of its longest proper border. Then:
\[
\mathit{next}_{-1}[i] = i + \lambda^{-1}[i] - \mathrm{border}(z).
\]
\end{lemma}

\begin{proof}
Write $z = bhb$, where $b \in \Sigma^*$ is the border, possibly empty, and $h \in \Sigma^+$. Then $x_i = bhb\alpha$ for some $\alpha \in \Sigma^+$. Let $j = \mathit{next}_{-1}[i]$. We claim that $x_j = b\alpha$, that is, the suffix $x_j$ starts at the second occurrence of the border.

Suppose first that $x_j$ starts before that occurrence, say $x_j = \beta b\alpha$ with $\beta \in \Sigma^+$ and $z = b\gamma\beta b$. Then $\beta b$ is a proper suffix of $z$, so by Definition~\ref{def:inv-lyn} we have $z \succ \beta b$. Hence $z\alpha \succ \beta b\alpha$, that is, $x_i \succ x_j$, contradicting the choice of $j$.

Suppose instead that $x_j$ starts after the border, say $x_j = \beta\alpha$ with $b = \gamma\beta$ and $\gamma,\beta \in \Sigma^+$. Consider the intermediate suffix $x_{k_m} = \beta h\gamma\beta\alpha$. Since $j = \mathit{next}_{-1}[i]$, we have $x_{k_m} \prec x_i \prec x_j$. However, from $x_j = \beta\alpha \succ \beta h\gamma\beta\alpha = x_{k_m}$, we obtain $b\alpha = \gamma\beta\alpha \succ \gamma\beta h\gamma\beta\alpha = x_i$, which implies $x_{k_m} \succ x_i$, again a contradiction.

Therefore $x_j = b\alpha$, and so
\[
j = i + |z| - |b| = i + \lambda^{-1}[i] - \mathrm{border}(z).
\]
\end{proof}

The border term is the main structural difference with the standard case. Standard Lyndon words are unbordered, so the corresponding correction term is always zero.

\begin{lemma}[Border via LCE]\label{lem:border-lce}
Let $z$ be the maximal inverse Lyndon subword starting at~$i$, and let $j = \mathit{next}_{-1}[i]$. Then
\[
\mathrm{lce}(i,j) = \mathrm{border}(z).
\]
\end{lemma}

\begin{proof}
Using the notation above, Lemma~\ref{lem:equiv-inv} yields $x_j = b\alpha$. Since $x_i = bhb\alpha$ and $x_j = b\alpha$, the suffixes $x_i$ and $x_j$ agree on the prefix $b$, so $\mathrm{lce}(i,j) \ge |b|$.

For the upper bound, note that $x_j \succ x_i$ since $j = \mathit{next}_{-1}[i]$. The two suffixes share a common prefix of length $|b|$ (both begin with $b$); at the next position $x_i[|b|{+}1] = h[1]$ while $x_j[|b|{+}1] = \alpha[1]$. Since $x_j \succ x_i$ and they agree on the first $|b|$ positions, the comparison is decided at position $|b|{+}1$, which forces $\alpha[1] > h[1]$. In particular $\alpha[1] \neq h[1]$, so the common prefix ends exactly at length $|b|$ and $\mathrm{lce}(i,j) = |b| = \mathrm{border}(z)$.
\end{proof}

\begin{figure}[t]
\centering
\begin{tikzpicture}[x=1.0cm,y=1.1cm,>=Latex,line cap=round,line join=round]

  \draw[thick] (0,0) rectangle (10.6,0.82);
  \fill[pattern=north east lines] (0,0) rectangle (2.0,0.82);
  \fill[pattern=horizontal lines] (2.0,0) rectangle (7.4,0.82);
  \fill[pattern=north east lines] (7.4,0) rectangle (9.4,0.82);
  \fill[pattern=dots] (9.4,0) rectangle (10.6,0.82);

  \draw[thick] (2.0,0) -- (2.0,0.82);
  \draw[thick] (7.4,0) -- (7.4,0.82);
  \draw[thick] (9.4,0) -- (9.4,0.82);

  \node at (1.0,0.41) {\footnotesize $b$};
  \node at (4.7,0.41) {\footnotesize $h$};
  \node at (8.4,0.41) {\footnotesize $b$};
  \node at (10.0,0.41) {\footnotesize $\alpha$};

  \draw[decorate,decoration={brace,amplitude=5pt},thick]
    (0,1.18) -- (9.4,1.18)
    node[midway,above=6pt] {\footnotesize maximal inverse Lyndon word $z=bhb$};

  \draw[decorate,decoration={brace,amplitude=4pt,mirror},thick]
    (0,-0.28) -- (10.6,-0.28)
    node[midway,below=5pt] {\footnotesize suffix $x_i$};

  \draw[decorate,decoration={brace,amplitude=4pt,mirror},thick]
    (7.4,-1.22) -- (10.6,-1.22)
    node[midway,below=5pt] {\footnotesize suffix $x_j$};

  \draw[->,thick] (0,-0.82) -- (0,-0.02);
  \node[font=\scriptsize,below] at (0,-0.82) {$i$};

  \draw[->,thick] (7.4,-0.82) -- (7.4,-0.02);
  \node[font=\scriptsize,below] at (7.4,-0.82) {$j=\mathit{next}_{-1}[i]$};

  \draw[decorate,decoration={brace,amplitude=4pt,mirror},thick]
    (0,-1.98) -- (2.0,-1.98)
    node[midway,below=5pt] {\footnotesize $\mathrm{lce}(i,j)=|b|$};

  \node[anchor=west,align=left,font=\footnotesize] at (11.05,0.35)
    {$\lambda^{-1}[i]$\\$=(j-i)+|b|$};

  \node[anchor=west,align=left,font=\scriptsize] at (11.05,-1.02)
    {start of the next greater suffix};

\end{tikzpicture}
\caption{Border correction for the inverse Lyndon array. If the maximal inverse Lyndon word starting at position $i$ has the form $z=bhb$, then the next greater suffix starts at the second occurrence of $b$. Consequently, $\mathrm{lce}(i,j)=|b|$ and $\lambda^{-1}[i]=(j-i)+|b|$.}
\label{fig:border-correction}
\end{figure}
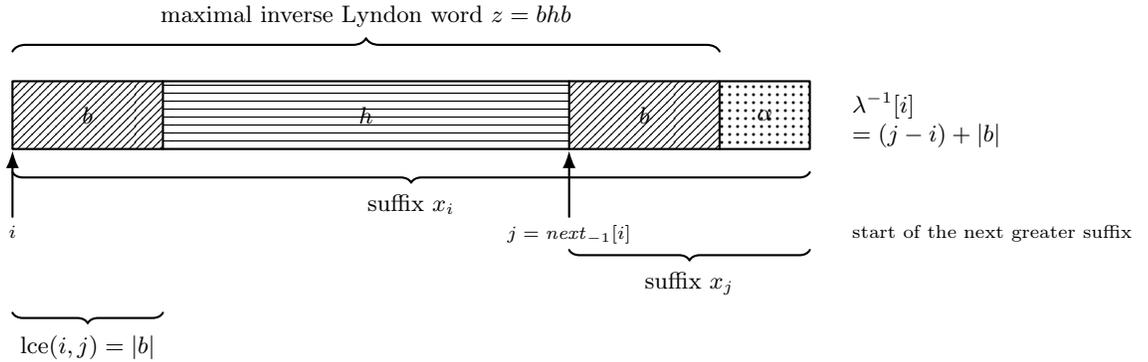

\begin{corollary}\label{cor:recover}
For every $i$,
\[
\lambda^{-1}[i] = \mathit{next}_{-1}[i] - i + \mathrm{lce}(i,\mathit{next}_{-1}[i]).
\]
\end{corollary}

\begin{remark}[Why the border term is structural]\label{rem:structural}
The correction term in Corollary~\ref{cor:recover} is not a post-processing detail. Without Lemma~\ref{lem:border-lce}, the inverse array would require explicit border information for each maximal inverse Lyndon factor. That would introduce an additional computational layer not present in the standard case. The point of Lemma~\ref{lem:border-lce} is precisely that the border is already encoded by the same nearest-suffix LCE information computed by the algorithm.
\end{remark}

Thus, computing $\mathit{next}_{-1}$ together with the associated $\mathit{nlce}$ array suffices to recover $\lambda^{-1}$.

\begin{lemma}\label{lem:prev-inv-lyndon}
For all $i \in [2,n]$,  $x[\mathit{prev}_{-1}[i]..i{-}1]$ is an inverse Lyndon word.
\end{lemma}

\begin{proof}
Let $p = \mathit{prev}_{-1}[i]$ and $w = x[p..i{-}1]$. By definition of PGS, $x_p \succ x_i$ and $x_k \prec x_i$ for every $k \in (p,i)$. By transitivity, $x_p \succ x_k$ for every $k \in (p,i)$.

Let $s = x[k..i{-}1]$ be any proper non-empty suffix of $w$, with $k \in (p,i)$. Since $x_p \succ x_k$, let $m \ge 1$ be the first index where $x_p$ and $x_k$ differ; then $x[p{+}m{-}1] > x[k{+}m{-}1]$. If $m \le i{-}k = |s|$, the mismatch falls within both $w$ and $s$, so $w \succ s$. If $m > i{-}k$, then $s = x[k..i{-}1]$ is a proper prefix of $w$ (they agree on all $|s|$ positions of $s$), so $w \succ s$ by the prefix rule of lexicographic order. In both cases $w \succ s$, so $w$ is an inverse Lyndon word.
\end{proof}

\subsection{Combinatorial Properties of NGS/PGS Edges}

We now show that the main structural lemmas from the NSS/PSS setting remain valid in the inverse case. These lemmas also fit into the broader context of non-crossing LCE-based techniques on general ordered alphabets. Here, however, the non-crossing structure is induced specifically by NGS/PGS edges in the inverse Lyndon setting.

\begin{lemma}[Non-Crossing of NGS/PGS Edges]\label{lem:inv-no-crossing}
Let $l_1 < r_1$ and $l_2 < r_2$ be index pairs, each connected by an NGS or PGS edge. Then it is impossible to have $l_1 < l_2 < r_1 < r_2$.
\end{lemma}

\begin{proof}
Assume for contradiction that $l_1 < l_2 < r_1 < r_2$. Since $l_2 \in (l_1,r_1)$ and $l_1,r_1$ are connected by an NGS or PGS edge, every suffix starting in $(l_1,r_1)$ is lexicographically smaller than both endpoints, hence $x_{l_2} \prec x_{r_1}$. On the other hand, since $r_1 \in (l_2,r_2)$ and $l_2,r_2$ are connected, we get $x_{l_2} \succ x_{r_1}$. This contradiction proves the claim.
\end{proof}

\begin{lemma}[Chain Iteration for NGS/PGS]\label{lem:inv-chain}
For any $r \in [2,n]$:
\begin{enumerate}
\item $\mathit{prev}_{-1}[r] = \mathit{prev}_{-1}^*[r{-}1]$;
\item $\mathit{next}_{-1}[l] = r$ if and only if $l = \mathit{prev}_{-1}^*[r{-}1]$ and $l > \mathit{prev}_{-1}[r]$.
\end{enumerate}
\end{lemma}

\begin{proof}
For the first statement, suppose $\mathit{prev}_{-1}[r] \neq \mathit{prev}_{-1}^*[r{-}1]$. Then there exists $r' = \mathit{prev}_{-1}^*[r{-}1]$ such that $\mathit{prev}_{-1}[r] \in (\mathit{prev}_{-1}[r'],r')$, giving
\[
\mathit{prev}_{-1}[r'] < \mathit{prev}_{-1}[r] < r' < r,
\]
which contradicts Lemma~\ref{lem:inv-no-crossing}.

For the second statement, assume first that $\mathit{next}_{-1}[l] = r$. Then all suffixes starting in $(l,r)$ are smaller than $x_r$, so $\mathit{prev}_{-1}[r] \notin [l,r)$ and therefore $\mathit{prev}_{-1}[r] < l$. If $l \neq \mathit{prev}_{-1}^*[r{-}1]$, then there exists $r' = \mathit{prev}_{-1}^*[r{-}1]$ such that
\[
\mathit{prev}_{-1}[r'] < l < r' < r,
\]
again contradicting Lemma~\ref{lem:inv-no-crossing}.

Conversely, assume that $\mathit{prev}_{-1}[r] < l$ and $l = \mathit{prev}_{-1}^*[r{-}1]$. Since $l \in (\mathit{prev}_{-1}[r],r)$, we have $x_l \prec x_r$. Moreover, the chain property implies $x_k \prec x_l$ for every $k \in (l,r)$. Therefore $r$ is the first position to the right of $l$ with a greater suffix, namely $\mathit{next}_{-1}[l] = r$.
\end{proof}

\begin{lemma}[LCE Acceleration for NGS/PGS]\label{lem:inv-lce-accel}
Let $l = \mathit{prev}_{-1}[k]$ and $r = \mathit{next}_{-1}[k]$. Then:
\begin{enumerate}
\item if $\mathrm{lce}(l,k) = \mathrm{lce}(k,r)$, then $\mathrm{lce}(l,r) \ge \mathrm{lce}(k,r)$, and either $\mathit{prev}_{-1}[r] = l$ or $\mathit{next}_{-1}[l] = r$;
\item if $\mathrm{lce}(l,k) < \mathrm{lce}(k,r)$, then $\mathrm{lce}(l,r) = \mathrm{lce}(l,k)$ and $\mathit{prev}_{-1}[r] = l$;
\item if $\mathrm{lce}(l,k) > \mathrm{lce}(k,r)$, then $\mathrm{lce}(l,r) = \mathrm{lce}(k,r)$ and $\mathit{next}_{-1}[l] = r$.
\end{enumerate}
\end{lemma}

\begin{proof}
For the first case, since $l = \mathit{prev}_{-1}[k]$ and $r = \mathit{next}_{-1}[k]$, we have $x_l \succ x_k \succ x_m$ and $x_r \succ x_k \succ x_m$ for all $m \in (l,r) \setminus \{k\}$. Since $\mathrm{lce}(l,k) = \mathrm{lce}(k,r)$, it follows that $\mathrm{lce}(l,r) \ge \mathrm{lce}(k,r)$. If $x_l \succ x_r$, then $x_l \succ x_r \succ x_m$ for all such $m$, so $\mathit{prev}_{-1}[r] = l$. By symmetry, if $x_r \succ x_l$, then $\mathit{next}_{-1}[l] = r$.

For the second case, let $c = \mathrm{lce}(l,k)$. Since $x_l \succ x_k$, we have $x[l+c] > x[k+c]$. Because $\mathrm{lce}(l,k) < \mathrm{lce}(k,r)$, the suffix $x_r$ shares a prefix of length at least $c+1$ with $x_k$, hence $x[r+c] = x[k+c]$. Thus $x_l$ and $x_r$ agree for $c$ positions and differ at the next one with $x[l+c] > x[r+c]$, so $x_l \succ x_r$ and $\mathrm{lce}(l,r) = c$. Since $x_r \succ x_m$ for all $m \in (l,r)$, it follows that $\mathit{prev}_{-1}[r] = l$.

The third case is symmetric, yielding $x_r \succ x_l$, $\mathrm{lce}(l,r) = \mathrm{lce}(k,r)$, and $\mathit{next}_{-1}[l] = r$.
\end{proof}

\subsection{The LCE-NGS Algorithm}

Lemmas~\ref{lem:inv-chain} and~\ref{lem:inv-lce-accel} yield Algorithm~\ref{alg:lce-ngs}, which mirrors the structure of the standard LCE-NSS algorithm~\cite{lyndon-simple}. The key difference is the direction of the comparison on line~8: in the inverse setting we test whether $x[\ell+m] < x[r+m]$, because we are looking for the first suffix to the right that is lexicographically greater.

\begin{algorithm}[t]
\caption{LCE-NGS for the inverse Lyndon array}\label{alg:lce-ngs}
\DontPrintSemicolon
\KwIn{word $x[1..n]$ in inverse sentinel mode}
\KwOut{$\mathit{next}_{-1}$, $\mathit{prev}_{-1}$, $\mathit{nlce}$, $\mathit{plce}$, $\lambda^{-1}$}
$\mathit{prev}_{-1}[1]\gets 0$; $\mathit{next}_{-1}[n]\gets n+1$; $\mathit{plce}[1]\gets 0$\;
\For{$r\gets 2$ \KwTo $n$}{
  $\ell\gets r-1$; $m\gets \textsc{Smart-Lce}(\ell,r,0)$\;
  \While{$x[\ell+m] < x[r+m]$}{
    $\mathit{next}_{-1}[\ell]\gets r$; $\mathit{nlce}[\ell]\gets m$\;
    \If{$m=\mathit{plce}[\ell]$}{$m\gets \textsc{Smart-Lce}(\mathit{prev}_{-1}[\ell],r,m)$\;}
    \ElseIf{$m>\mathit{plce}[\ell]$}{$m\gets \mathit{plce}[\ell]$\;}
    $\ell\gets \mathit{prev}_{-1}[\ell]$\;
  }
  $\mathit{prev}_{-1}[r]\gets \ell$; $\mathit{plce}[r]\gets m$\;
}
\For{$i\gets 1$ \KwTo $n$}{$\lambda^{-1}[i]\gets \mathit{next}_{-1}[i]-i+\mathit{nlce}[i]$\;}
\end{algorithm}

The procedure \textsc{Smart-Lce} is the same amortized primitive used in the standard setting~\cite{lyndon-simple}, but here it is queried only on NGS/PGS candidates. For self-containment, we recall the two invariants that are used in the proof of Theorem~\ref{thm:linear}. First, after each explicit scan, $\mathit{rhs}$ is the rightmost text position inspected so far. Second, whenever an edge value has already been fixed, the corresponding stored quantity in $\mathit{nlce}$ or $\mathit{plce}$ is exactly the required LCE for that edge. Hence, if a query $(\ell,r,m)$ satisfies $r+m<\mathit{rhs}$, then the answer lies entirely inside a previously certified window and can be returned in $O(1)$ time from the stored edge information selected by Lemma~\ref{lem:inv-lce-accel}. Otherwise, \textsc{Smart-Lce} extends explicitly, increases $\mathit{rhs}$, and stores the resulting LCE value. This is the only place where new character comparisons are performed.

\begin{theorem}\label{thm:linear}
Algorithm~\ref{alg:lce-ngs} computes the inverse Lyndon array $\lambda^{-1}[1..n]$ in $O(n)$ time and $O(n)$ space on general ordered alphabets.
\end{theorem}

\begin{proof}
Correctness follows from Lemmas~\ref{lem:equiv-inv}, \ref{lem:border-lce}, \ref{lem:inv-no-crossing}, \ref{lem:inv-chain}, and~\ref{lem:inv-lce-accel}. The additional point to justify is that the border correction does not create extra rescanning.

By Lemma~\ref{lem:border-lce}, whenever $\mathit{next}_{-1}[\ell]=r$ is fixed, the same value $\mathit{nlce}[\ell]=\mathrm{lce}(\ell,r)$ is both the edge LCE needed by the algorithm and the border correction needed later for $\lambda^{-1}[\ell]$. Therefore the algorithm never computes borders explicitly and never performs a separate pass for border information.

Moreover, if \textsc{Smart-Lce} is called with $r+m<\mathit{rhs}$, then the query lies entirely inside a region already certified by a previous explicit scan. In that case, Lemma~\ref{lem:inv-lce-accel} and the stored $\mathit{nlce}$ or $\mathit{plce}$ values suffice to answer the query without rescanning characters inside that window. Thus borders cannot trigger hidden rescanning there. Every explicit character comparison strictly advances the global frontier $\mathit{rhs}$, which is monotone and never exceeds $n$. Hence the total number of explicit comparisons is $O(n)$.

The outer loop performs $n-1$ iterations. Each index receives its final $\mathit{next}_{-1}$ value once, each stored LCE value is written once, and the final recovery formula
\[
\lambda^{-1}[i] = \mathit{next}_{-1}[i]-i+\mathit{nlce}[i]
\]
is applied once per position. Therefore the total running time is $O(n)$ and the space usage is $O(n)$.
\end{proof}

\section{Experimental Evaluation}\label{sec:experiments}

We implemented LCE-NSS for $\lambda$ and LCE-NGS for $\lambda^{-1}$ in C++17. The timing runs reported here were produced by the benchmark pipeline with \texttt{g++ -std=c++17 -O2}. We measured wall-clock construction time and, when relevant, report the recovery pass separately. All experiments were run on a Dell Precision 7960 Tower under Ubuntu Linux, with an Intel Xeon w5-3425 CPU and 128\,GB of ECC RAM. Correctness was checked against brute force on small and medium inputs and on crafted edge cases.

We used random words over alphabets of size $\sigma\in\{2,4,26\}$, structured synthetic families, and real texts from Pizza\&Chili and the Large Canterbury Corpus~\cite{pizza-chili,canterbury-corpus,canterbury-desc}. The profiling-enabled implementation also counted explicit character comparisons, LCE reuse hits inside the certified \textsc{Smart-Lce} window, and explicit extension calls.

Table~\ref{tab:results} reports the random-input timings. On instances with $n\ge 5\times 10^4$, the mean ratio $\textsc{LCE-NGS}/\textsc{LCE-NSS}$ is $1.0060$, the median ratio is $0.9990$, and the observed range is $[0.9882,1.0939]$, which shows that the inverse construction matches the standard one closely across all alphabet sizes.

\begin{table}[t]
\caption{Core construction time ($\mu$s) on random words.}
\label{tab:results}
\centering
\resizebox{\textwidth}{!}{\begin{tabular}{@{}r@{\quad}rr@{\qquad}rr@{\qquad}rr@{}}
\toprule
 & \multicolumn{2}{c}{$\sigma = 2$} & \multicolumn{2}{c}{$\sigma = 4$} & \multicolumn{2}{c}{$\sigma = 26$} \\
\cmidrule(lr){2-3}\cmidrule(lr){4-5}\cmidrule(lr){6-7}
$n$ & LCE-NSS & LCE-NGS & LCE-NSS & LCE-NGS & LCE-NSS & LCE-NGS \\
\midrule
$10^3$ & 15.4 & 15.0 & 15.8 & 15.4 & 15.6 & 16.0 \\
$5\times 10^3$ & 79.6 & 84.2 & 95.6 & 101.2 & 94.0 & 92.6 \\
$10^4$ & 171.4 & 179.0 & 205.4 & 207.0 & 194.0 & 190.6 \\
$5\!\times\!10^4$ & 948.2 & 1\,037.2 & 1\,062.0 & 1\,059.2 & 1\,006.6 & 1\,000.6 \\
$10^5$ & 1\,974.2 & 1\,973.4 & 2\,187.0 & 2\,162.4 & 2\,060.6 & 2\,053.8 \\
$5\!\times\!10^5$ & 9\,197.8 & 9\,416.2 & 10\,503.2 & 10\,543.2 & 10\,087.8 & 9\,981.4 \\
$10^6$ & 18\,060.0 & 18\,320.2 & 20\,972.4 & 20\,955.0 & 20\,121.8 & 19\,919.2 \\
$2\!\times\!10^6$ & 35\,778.4 & 36\,154.0 & 41\,893.6 & 41\,840.6 & 40\,219.0 & 39\,744.4 \\
$5\!\times\!10^6$ & 89\,178.6 & 90\,033.6 & 102\,314.0 & 104\,479.2 & 100\,569.0 & 99\,436.6 \\
\bottomrule
\end{tabular}}
\end{table}

Table~\ref{tab:real-results} reports the results obtained on real files. Across the five corpora, the mean ratio is $1.0019$, the median ratio is $1.0024$, and the observed range is $[0.9665,1.0325]$. The structured synthetic families show the same practical behavior. Over all structured instances with $n\ge 5\times 10^4$, the mean ratio is $0.9852$, the median ratio is $0.9981$, and the observed range is $[0.8848,1.0874]$. Thus the inverse kernel remains close to the standard one across random, structured, and real inputs.

\begin{table}[t]
\caption{Core construction time ($\mu$s) on real corpora.}
\label{tab:real-results}
\centering
\begin{tabular}{@{}lrrrr@{}}
\toprule
File & $n$ & LCE-NSS & LCE-NGS & Ratio \\
\midrule
\texttt{english.txt} & 5\,000\,000 & 98\,850.6 & 98\,846.6 & 1.0000 \\
\texttt{dna.txt} & 5\,000\,000 & 104\,074.2 & 104\,328.6 & 1.0024 \\
\texttt{bible.txt} & 4\,047\,392 & 75\,372.2 & 77\,819.2 & 1.0325 \\
\texttt{e.coli} & 4\,638\,690 & 96\,704.0 & 93\,460.0 & 0.9665 \\
\texttt{world192.txt} & 2\,473\,400 & 46\,763.6 & 47\,144.6 & 1.0081 \\
\bottomrule
\end{tabular}
\end{table}

To stress the border-correction path, we additionally profiled long-border families while recording explicit character comparisons, LCE reuse hits, and explicit extension calls inside \textsc{Smart-Lce}. Each instance contains a repeated prefix-suffix border occupying either $25\%$ or $40\%$ of the full length. The timing ratios on these border-heavy inputs remain close to one, mean $1.0013$, median $1.0031$, and range $[0.9680,1.0257]$. Table~\ref{tab:border-counters} shows that all counters remain linear in $n$. In particular, from $10^5$ to $10^6$, both explicit comparison counts and explicit extension counts grow by about $10\times$ in both families, which matches Theorem~\ref{thm:linear}. Deep borders do affect constants, but they do not trigger superlinear rescanning.

\begin{table}[t]
\caption{Border-heavy profiling counters. Values are averages over three runs.}
\label{tab:border-counters}
\centering
\scriptsize
\begin{tabular}{@{}llrrrrrr@{}}
\toprule
Family & $n$ & NSS cmp. & NGS cmp. & NSS reuse & NGS reuse & NSS ext. & NGS ext. \\
\midrule
25\% border & 100\,000    &   213\,984 &   213\,427 &  29\,347 &  29\,258 &   148\,890 &   148\,233 \\
40\% border & 100\,000    &   200\,510 &   208\,599 &  44\,719 &  35\,317 &   131\,483 &   141\,882 \\
25\% border & 500\,000    &   945\,829 & 1\,039\,872 & 294\,036 & 179\,692 &   582\,948 &   705\,698 \\
40\% border & 500\,000    &   947\,176 &   863\,646 & 292\,529 & 393\,730 &   584\,150 &   475\,202 \\
25\% border & 1\,000\,000 & 2\,134\,164 & 2\,134\,660 & 293\,205 & 293\,234 & 1\,481\,797 & 1\,483\,100 \\
40\% border & 1\,000\,000 & 2\,135\,385 & 2\,123\,841 & 292\,240 & 306\,666 & 1\,483\,741 & 1\,468\,012 \\
\bottomrule
\end{tabular}
\end{table}

\section{Conclusions and Suffix-Sorting Perspectives}\label{sec:suffix}

We introduced the inverse Lyndon array and characterized it in terms of the next greater suffix array and a border correction term. We also showed that the combinatorial ingredients behind the NSS/PSS framework transfer to the NGS/PGS setting, yielding a linear-time construction on general ordered alphabets. The central technical point is that borders do not create hidden rescanning costs: the required correction term is already encoded by LCE values on nearest-greater-suffix edges and is absorbed by the same amortized machinery as in the standard case. Experimentally, the inverse construction preserves the same practical linear-time profile as the standard one.

A natural perspective is suffix sorting with the aid of the standard and inverse Lyndon arrays. It is well known that the standard Lyndon array supports suffix sorting through the compatibility property
\[
w_\lambda(i) \prec w_\lambda(j) \Longrightarrow x_i \prec x_j,
\]
where $w_\lambda(i)$ denotes the maximal Lyndon subword starting at~$i$~\cite{baier:LIPIcs.CPM.2016.23,accelerate-suffix-sorting}.

By contrast, the inverse case does not admit a direct analogue.

\begin{proposition}[No direct inverse compatibility]\label{prop:no-direct-inverse}
There is no direct inverse analogue of the standard compatibility property for maximal inverse Lyndon factors. In particular, it is not true in general that
\[
w_{\lambda^{-1}}(i) \prec w_{\lambda^{-1}}(j) \Longrightarrow x_i \prec x_j.
\]
\end{proposition}

\begin{proof}
Take $x=\mathit{babacbabaa}$ with $i=1$ and $j=6$. Then $w_{\lambda^{-1}}(1)=\mathit{baba}$ and $w_{\lambda^{-1}}(6)=\mathit{babaa}$, so $w_{\lambda^{-1}}(1)\prec w_{\lambda^{-1}}(6)$, but $x_1=\mathit{babacbabaa}\succ x_6=\mathit{babaa}$. The obstruction is the prefix case: if an inverse Lyndon word has the form $z=z'\alpha$, one may have $z\succ \alpha$ while only $z'\succeq \alpha$, so the continuation of the suffix is not determined by the maximal inverse factors alone.
\end{proof}

Nevertheless, the inverse Lyndon array still provides constant-time comparison rules for relevant suffix pairs.

\begin{proposition}\label{prop:short-circuit}
For positions $i<j$, each of the following conditions determines the order of $x_i$ and $x_j$ in constant time:
\begin{enumerate}
\item if $j < i + \lambda[i]$, then $x_i \prec x_j$;
\item if $j = \mathit{next}_{-1}[i]$, equivalently $j = i + \lambda^{-1}[i] - \mathit{nlce}[i]$, then $x_i \prec x_j$;
\item if $j < \mathit{next}_{-1}[i]$, equivalently $j < i + \lambda^{-1}[i] - \mathit{nlce}[i]$, then $x_i \succ x_j$;
\item if $i = \mathit{prev}[j]$ or $i = \mathit{prev}_{-1}[j]$, the comparison follows immediately from the defining edge.
\end{enumerate}
\end{proposition}

\begin{proof}
Each item follows directly from the definitions.
\end{proof}

Combining the present results with other properties of ICFL may support efficient approaches to suffix sorting, a direction that also aligns with the broader formal-language perspective discussed in~\cite{DLT22}.

\begin{credits}
\subsubsection{\discintname}
The authors have no competing interests to declare that are relevant to the content of this article.
\end{credits}

\vspace*{\prebibadjust}

\end{document}